\documentclass[11pt]{article} 
\usepackage[dvips]{graphicx}                 
\usepackage{fullpage} 
\usepackage{verbatim} 
\usepackage{amsmath}
\usepackage{times}
\usepackage{hyperref}
\usepackage{algorithm}
\usepackage{algorithmic}
\usepackage{enumerate}
\usepackage[font=small]{caption}
\usepackage{multirow}

\newtheorem{lemma}{Lemma}
\newtheorem{corollary}{Corollary}
\newenvironment{proof}{{\noindent \textbf{Proof:}}}{\hfill\rule{2mm}{2mm}\par}

\title{Engineering Small Space Dictionary Matching}
\author{Shoshana Marcus\thanks{
Simons Center for Quantitative Biology,
Cold Spring Harbor Laboratory,
1 Bungtown Road, Cold Spring Harbor, N.Y., 11724.
email: {\tt smarcus@cshl.edu}.
}
\and Dina Sokol \thanks{
Department of Computer and Information Science,
Brooklyn College of the City University of New York, 2900 Bedford Avenue,
Brooklyn, N.Y. 11210. email: {\tt sokol@sci.brooklyn.cuny.edu}.
}
}
\begin{document}
\date{}
\maketitle

\begin{abstract}
The dictionary matching problem is to locate occurrences of any pattern among a set of patterns in a given text.
Massive data sets abound and at the same time, there are many settings in which working space is extremely limited.  
We introduce dictionary matching software for the space-constrained environment whose running time is close to linear.
We use the compressed suffix tree as the underlying data structure of our algorithm, thus,
the working space of our algorithm is proportional to the optimal compression of the dictionary.  
We also contribute a succinct tool for performing constant-time lowest marked ancestor queries 
on a tree that is succinctly encoded as a sequence of balanced parentheses, with linear time preprocessing of the tree.  
This tool should be useful in many other applications.   \\
Our source code is available at \urlstyle{same}\url{http://www.sci.brooklyn.cuny.edu/~sokol/dictmatch.html}

\end{abstract}

\section{Introduction}
In recent years, there has been a massive proliferation of digital data. 
Concurrently, industry has been producing equipment
with ever-decreasing hardware availability. Thus, we are faced with scenarios
in which this data growth must be accessible to applications running on devices that have
reduced storage capacity, such as mobile and satellite devices. Hardware resources are
more limited, yet the user's expectations of software capability continue to escalate.
This unprecedented rate of digital data accumulation therefore presents a constant challenge
to the algorithms and software developers who must work with a shrinking hardware
capacity.

The dictionary matching problem is to identify a \textit{set} of patterns, called a dictionary, within a given text.
Applications for this problem include searching for specific phrases in a book, scanning a file
for virus signatures, and network intrusion detection. The problem also has applications
in the biological sciences, such as searching through a DNA sequence for a set of motifs, identifying motifs to characterize protein families, 
and finding anchors for fast alignment of large genomic sequences.


A series of dictionary matching algorithms that operate in small space have in fact been developed \cite{ChanHLS07, HLSTVit08, Bel10, HonKSTV10}. 
The latest of these results \cite{HonKSTV10} achieved time and space optimal 1D dictionary matching.  That is, their algorithm runs in linear time within space that meets empirical entropy bounds of the dictionary.  
The empirical entropy of a string ($H_0$ or $H_k$) describes the minimum number of bits that are needed to encode the string within context.  

Succinct  dictionary matching algorithms have remained in the theoretical realm and have not been implemented until now.  This work fills the void.  
We have developed software for dictionary matching in small space that relies on the compressed suffix tree, a popular succinct data structure. 
Our main challenge lay in combining a dictionary matching algorithm for the generalized suffix tree with compressed suffix tree representations.  

Chan et al.\ developed the first succinct dictionary matching algorithm \cite{ChanHLS07}.  Their algorithm uses the compressed suffix tree of Sadakane \cite{Sad07}, which they extended so that it can support a dynamically changing dictionary of patterns.
Hon et al.\ presented a more space-efficient  dictionary matching algorithm that uses a sampling technique to compress a suffix tree \cite{HLSTVit08}.
Their algorithm uses several data structures
along with a compressed representation of the suffix tree, among them the string B-tree, a compressed trie of the patterns, 
and an LCP array of the longest common prefixes between the patterns.
Our software uses only the compressed suffix tree, augmented with a succinct framework for lowest marked ancestor queries.

We also contribute a succinct tool for performing lowest marked ancestor queries in constant time after linear time preprocessing of a compressed suffix tree. 
The compressed suffix tree is augmented by a bit vector and a sequence of balanced parentheses.  Lowest marked ancestor queries are answered by a set of constant-time queries to these data structures.
The lowest marked ancestor structure that we implemented  is appropriate for any succinct representation of an ordered tree that encodes the structure as a sequence of balanced parentheses, which was introduced by Jacobson \cite{Jac89}.  
Thus, our tool for lowest marked ancestor queries is a contribution that is useful to other applications that use the balanced parentheses representation of a tree.

We begin with  an overview of  compressed suffix trees in Section \ref{sec:CST}, since they are the basis of our succinct dictionary matching software.  
In Section \ref{sec:softwareLin} we describe our linear-time dictionary matching software that relies on the uncompressed suffix tree.
Then, in Section \ref{sec:software}, we describe the techniques employed by our succinct dictionary matching software, and the succinct framework we implemented for lowest marked ancestor queries on a compressed suffix tree.  In Section \ref{sec:experiments}, we present experimental results to demonstrate that these techniques are in fact space-efficient.  We conclude with a summary and direction for future work in Section \ref{sec:conclusion}.

\section{Compressed Suffix Tree}\label{sec:CST}

We begin with a description of the suffix tree and of compressed representations of the suffix tree since our program relies on the compressed suffix tree as its underlying data structure.
The suffix tree is a compact trie that represents all suffixes of an input string. 
The suffix tree for $S=s_1s_2 \cdots s_n$ is a rooted, directed tree with $n$ leaves, one for each suffix.
Each internal node, except the root, has at least two children.	 Each edge is labeled with a nonempty substring of $S$ and no two edges emanating from a node begin with the same character.  The	path from the root to leaf $i$ spells out suffix $S[i \ldots n]$.
Suffix links allow an algorithm to move quickly to a distant part of the tree.  A suffix link is a pointer from an internal node labeled $x\alpha$ to another internal node labeled $\alpha$, where $x$ is an arbitrary character and $\alpha$ is a possibly empty substring.  
The suffix array is a data structure that indexes a string by storing the lexicographical order of its suffixes.  

Recent innovations in succinct full-text indexing provide us with the ability to compress both a suffix array and a suffix tree,
using space that is proportional to the optimal compression of the data they are built upon.
These self-indexes can replace the original text, as they support retrieval of the original text,
in addition to answering queries about the data very quickly. 
Several compressed suffix array (CSA) representations exist, e.g., \cite{Sad03, GroGupVit03, FerMan05, FerManMakNav07}, each with a different time-space trade-off. The most recent results meet $k$th order empirical entropy of the input string.  
Compressed representations of the suffix tree, e.g., \cite{Sad07,  RusNavOli11, Fis10, OhlFisGog10}, use the compressed suffix array as a component. Table \ref{table:CST} summarizes the time-space trade-offs in several compressed suffix tree (CST) representations.



\begin{table}
	\centering
		\begin{tabular}{|l|l|l|}
			\hline
			\textbf{Space (bits)} & \textbf{Slowdown} & \textbf{Reference}\\
			\hline
			$O(\ell \log \ell)$ & $O(1)$ & Uncompressed suffix tree \\
			\hline
			$O(\ell \log \sigma)$ & $O(polylog(\ell))$ & Sadakane \cite{Sad07} \\
			\hline
			$\ell H_k(T)+o(\ell \log \sigma)$ & $O(\log \ell)$ &  Russo et al.\ \cite{RusNavOli11} \\ 
			\hline
			$(1+\frac{1}{\epsilon})\ell H_k(T)+o(\ell \log \sigma)$ & $O(\log^\epsilon \: \ell)$, $0<\epsilon \leq 1$ 
 				& Fischer et al.\ \cite{Fis10} \\
			\hline
			$|CSA|+|CLCP|+3\ell$ & $O(1)$ for many operations & Ohlebusch et al.\ \cite{OhlFisGog10} \\			
			\hline
		\end{tabular}
		\caption{Compressed suffix tree representations for an input string $T$ of length $\ell$, where $|CSA|$ is the number of bits used to store the compressed suffix array and $|CLCP|$ is the number of bits occupied by the compressed LCP array.} 
		\label{table:CST}
\end{table}

Ohlebusch et al.\ \cite{OhlFisGog10} recognized that the compressed suffix tree generally consists of  three separate parts: the
lexicographical information in a compressed suffix array (CSA), the information about
common substrings in the longest common prefix array (LCP), and the tree topology
combined with a navigational structure (NAV). Each of these three components functions independently from the others and is stored separately.  
Representations of compressed suffix arrays and compressed LCP arrays are interchangeable in many compressed suffix tree representations. Combining the different representations of each component yields a rich variety of compressed suffix trees, although some compressed suffix trees favor certain compressed suffix array or compressed LCP array representations.  The Succinct Data Structures Library (SDSL) provides a range of compressed suffix tree implementations, which we used in our  dictionary matching software. We experimented with Sadakane's compressed suffix tree \cite{Sad07} by using an assortment of compressed suffix array and compressed LCP modules to achieve different time and space complexities in our dictionary matching software.

\section{Linear-Time Dictionary Matching with Suffix Tree}\label{sec:softwareLin}
We first developed a linear-time dictionary matching program that uses the uncompressed suffix tree as its primary data structure.  Then we modified our approach to use the compressed suffix tree to improve the space complexity.
In this section we describe the linear time dictionary matching software that uses an uncompressed suffix tree.
Then, in the next section, we delineate the revisions in our techniques so that we perform dictionary matching using compressed suffix tree representations.
 
A suffix tree can be used to index several strings, in a \emph{generalized suffix tree}.
The dictionary can be merged to form a single string by concatenating the  patterns with a unique delimiter separating them.  
Because it is \textit{online}, Ukkonen's suffix tree construction algorithm can insert one string at a time and index only the actual suffixes of a set of strings in a suffix tree \cite{Gusfield97}.
The dictionary of patterns  is indexed by a generalized suffix tree to preprocess it for dictionary matching queries.  
Then, the text is searched for pattern occurrences in linear time, in a manner similar to Ukkonen's insertion of a new string to the suffix tree.  
We briefly summarize Ukknonen's suffix tree construction algorithm in the following paragraph, and depict its steps in Algorithm \ref{alg:UkkonenMain}.

	\begin{algorithm}
	\small
	\caption{Ukkonen's suffix tree construction algorithm}
	\label{alg:UkkonenMain}
	\begin{algorithmic}
		\STATE j = -1; \STATE \COMMENT{$j$ is last suffix inserted}
		\FOR{$i = 0$ to $n-1$} \STATE \COMMENT{phase $i$: $i$ is current end of string} 
				\WHILE{$j<i$} \STATE \COMMENT{let $j$ catch up to $i$}
					\IF{singleExtensionAlgorithm(i, j)} 
							\STATE break \COMMENT{implicit suffix so proceed to next phase}
						\ENDIF
					
						\IF{$lastNodeInserted \neq root$}
						 	\STATE $lastNodeInserted.SuffixLink \gets root$
						\ENDIF
					\STATE $lastNodeInserted \gets root$
				\ENDWHILE
		\ENDFOR
	\end{algorithmic}
	\end{algorithm} 
	\normalsize

 The elegance of Ukkonen's algorithm is evident in its key property.  The algorithm admits the arrival of the string during construction.  Yet, each suffix is inserted exactly once, and a leaf is never updated after its creation.
As a new character is appended to the input string, Ukkonens's algorithm ensures that all suffixes of the string are indexed by the tree.  As soon as a suffix is implicitly found in the tree, modification of the tree ceases until the next new character is examined.  The next phase begins by extending the implicit suffix with the new character.  
Using suffix links and a pointer to the last suffix inserted, each suffix is inserted in the tree in amortized constant time.  The combination of one-time insertion of each suffix and rapid suffix insertion  results in an overall linear-time suffix tree construction algorithm.
 
Our dictionary matching algorithm over a generalized suffix tree of patterns was inspired by Ukkonen's process for inserting a new string into a generalized suffix tree (as shown in Algorithm \ref{alg:UkkonenMain}), pretending to index the text, without modifying the index.  
The text is processed in an online fashion, traversing the suffix tree of patterns as each successive character of text is processed.  A pattern occurrence is announced when a labeled leaf is encountered, i.e., a leaf that represents the first suffix of a pattern.
At a position of mismatch and at a pattern occurrence, suffix links are used to navigate to successively smaller suffixes of the matching string.  
When a suffix link is used within the label of a node, the corresponding number of characters can be skipped, obviating redundant character comparisons. In the spirit of Ukkonen's skip-count trick, this ensures that the text is scanned in linear time.
   
The skip-count trick is based on  Lemma \ref{lemma:numNodes}.  A suffix link is  a directed edge from the internal node at the end of the path labeled $x \alpha$ to another internal node at the end of the path labeled $\alpha$, where $x$ is an arbitrary character and $\alpha$ is a possibly empty substring.  We can similarly define suffix links for leaves in the tree.  The suffix link of the leaf representing suffix $i$ points to the leaf representing suffix $i+1$.

\begin{lemma}\label{lemma:numNodes}\cite{Gusfield97}
In a suffix tree, the number of nodes along the path labeled $\alpha$ is at least as many as the number of nodes along the path labeled $x\alpha$.
\end{lemma}
\begin{proof}
Suppose not.  That is, there is some $\alpha$ for which the path labeled $\alpha$  has fewer nodes than the path labeled $x\alpha$.  This means that some suffix of $x \alpha$ is not indexed by the suffix tree.   This implies that the suffix tree is not fully constructed.  Hence, a contradiction and the premise must be valid.
\end{proof}

\begin{corollary}
If the suffix link of the root points to itself, every node of the suffix tree has a suffix link.
\end{corollary}

Ukkonen uses suffix links to navigate across a suffix tree and then skip over the appropriate number of characters labeling the beginning of an edge.  In dictionary matching, we navigate a fully constructed suffix tree, and every node must have a suffix link established for it.
  To avoid redundant comparisons, we follow a suffix link across a suffix tree and then jump up to the position at which the mismatch occurred.  It is more efficient to navigate up a suffix tree than down.  That is, every node has a single parent but when navigating to a child, several branches can be considered. 
  When a suffix link is traversed, we know the number of characters to skip going up the edge as this is the number of characters that remain along the edge  after the position of mismatch.  Yet, we do not know at which node traversal will halt.  The shorter label may be split over more edges than the longer label spans, by Lemma \ref{lemma:numNodes}.  

We extended Algorithm \ref{alg:UkkonenMain} to perform dictionary matching.  Pseudocode of our program is delineated in Algorithm \ref{alg:DictST}, with its submodules extracted to Algorithms  \ref{alg:Mismatch} and \ref{alg:SkipCount}.
Our program reports the longest pattern occurrence that ends at each text position.
When a pattern is a suffix of a longer pattern, and the longer pattern occurs in the text, we do not spend time reporting an occurrence of the shorter pattern.

	\begin{algorithm}
	\small
	\caption{Dictionary matching over the generalized suffix tree }
	\label{alg:DictST}
	\begin{algorithmic}[1]
		\STATE \textit{curNode} $\gets$ \textit{root}
		\STATE \textit{textIndex} $\gets 0$ 
		\STATE \textit{curNodeIndex} $\gets 0$ 
		\STATE \textit{skipcount} $\gets 0$
		\STATE \textit{usedSkipcount} $\gets false$
		\REPEAT
      \STATE \textit{lastNode} $\gets$ \textit{curNode}
			\IF {\textit{usedSkipCount} $\neq$ \textit{true}}		
				\STATE \textit{textIndex}$+=$\textit{curNodeIndex} 
				\STATE \textit{curNodeIndex} $\gets 0$

				\STATE \textit{curNode} $\gets$ \textit{curNode.child(text[textIndex])}

        \IF {\textit{curNode.length}$>0$}
        	\STATE \textit{curNodeIndex}$++$ \COMMENT{already compared the first character on the edge}
        \ENDIF
			\ELSE
    		\STATE \textit{usedSkipCount} $\gets$ \textit{false}
			\ENDIF

				\STATE \COMMENT {compare text}
				\WHILE{\textit{curNodeIndex}$<$\textit{curNode.length} AND \textit{curNodeIndex}$+$\textit{textIndex}$<$\textit{text.length}}
					\IF {$text[textIndex+curNodeIndex]\neq pat[curNode.stringNum][curNode.beg+curNodeIndex]$}					
						\STATE break \COMMENT{mismatch}
					\ENDIF
					\STATE \textit{curNodeIndex}$++$
				\ENDWHILE

					\IF{\textit{curNodeIndex=curNode.length} AND \textit{curNode.firstLeaf()}}\label{line:announce}
						\STATE {announce pattern occurrence} 
					\ENDIF

				\IF{$curNodeIndex = curNode.length$ AND $curNode.length>0$ AND $text[textIndex+curNodeIndex-1]=pat[curNode.stringNum][curNode.beg+curNodeIndex-1]$}
        			\STATE continue  \COMMENT{branch and continue comparing text to patterns}
				\ENDIF
				\STATE {handleMismatch}
		\UNTIL{\textit{textIndex+curNodeIndex} $\geq$ \textit{text.length}} \COMMENT{scan entire text}				

	\end{algorithmic}
	\end{algorithm} 
	\normalsize

	\begin{algorithm}
	\small
	\caption{Handling a Mismatch}
	\label{alg:Mismatch}
	\begin{algorithmic}

					\IF{\textit{curNode.depth} $\neq 0$ OR \textit{lastNode.depth} $\neq 0$} 
						\IF{\textit{curNode.suffixLink = root} AND  \textit{lastNode.suffixLink}$\neq$ \textit{root}}
							    \STATE \textit{curNode} $\gets$ \textit{lastNode}
							    \STATE \textit{curNodeIndex} $\gets$ \textit{curNode.length}  \COMMENT{mismatched when trying to branch}
							    \STATE \textit{textIndex} $-=$ \textit{curNode.length}
						\ENDIF

		\IF{ \textit{curNode.parent} $=$ \textit{root} AND \textit{curNodeIndex} $= 1$}   
	   		\STATE \textit{textIndex}$++$
	       \STATE \textit{curNodeIndex} $= 0$
	       \STATE \textit{curNode = curNode.parent} 
	       \STATE continue \COMMENT{when traverse suffix link: will be at mismatch, so skip 1 char}
		\ENDIF

		\STATE useSkipcountTrick(skipcount, curNode) 

					\ELSE \STATE \COMMENT{mismatch at root}
							\STATE \textit{textIndex++}
           \ENDIF

	\end{algorithmic}
	\end{algorithm} 
	\normalsize

	\begin{algorithm}
	\small
	\caption{Skip-Count trick}
	\label{alg:SkipCount}
	\begin{algorithmic}

\REPEAT 
\STATE \textit{curNode} $\gets$ \textit{curNode.suffixLink}
\STATE \textit{usedSkipCount} $\gets$ \textit{true}
\STATE \textit{textPos} $=$ \textit{curNodeIndex+textIndex}
\STATE \textit{skipcount} $\gets$ \textit{curNode.length} $-$ \textit{curNodeIndex}

\IF {\textit{skipcount} $\geq$ \textit{curNode.length}}

	\IF{\textit{curNode.length} $= 0$}
		\STATE{\textit{usedSkipCount} $\gets$ \textit{false}} \COMMENT{branch at next iteration of outer loop, look for next text char}
		\STATE{\textit{curNodeIndex} $\gets 0$}
		\STATE{\textit{skipcount} $\gets 0$}
	\ELSE
		\IF{\textit{skipcount} $=$ \textit{curNode.length}}
    	\STATE{\textit{curNodeIndex}$--$}
      \STATE{\textit{usedSkipCount} $\gets$ \textit{false}} \COMMENT{branch at next iteration of outer loop}                   
     \ENDIF
     \STATE \textit{skipcount} $-=$ \textit{curNode.length}
     \STATE \textit{curNode} $\gets$ \textit{curNode.parent}
	\ENDIF
	\ELSE
   	\STATE \textit{curNodeIndex} $\gets$ \textit{curNode.length} $-$ \textit{skipcount}
		\STATE \textit{skipcount} $\gets 0$	
\ENDIF

		\UNTIL{\textit{skipcount} $\leq 0$} 
\STATE \textit{textIndex} $=$ \textit{textPos} $-$ \textit{curNodeIndex}

	\end{algorithmic}
	\end{algorithm}

A key challenge in implementing dictionary matching on the suffix tree is the scenario in which one  pattern is a proper substring of another pattern \cite{AmiFar91}. Traversing the suffix tree using suffix links (as in Algorithm \ref{alg:DictST}), these pattern occurrences can be passed unnoticed in the text. This limitation is addressed by augmenting each node of the suffix tree with a pointer to the longest prefix of the label along its path from the root that is a complete pattern.  
The nodes are marked with this information in linear time by a depth-first traversal of the suffix tree.  

The suffix tree is a versatile tool in string algorithms, and is already needed in many applications to facilitate
other queries.  Thus, in practice, our linear-time dictionary matching program with the uncompressed suffix tree requires very little additional space.  
This tool is itself a contribution, allowing efficient dictionary matching in small space, however, we improved this application by using a compressed suffix tree as the underlying data structure.

\section{Dictionary Matching with Compressed Suffix Tree}\label{sec:software}

In this section we describe how we redesigned our dictionary matching code to run over a compressed suffix tree in linear time, overlooking the slowdown of queries on the compressed suffix tree.
Since the existing compressed suffix suffix tree construction algorithms are not online algorithms, it is not possible to build the compressed suffix tree incrementally, inserting one pattern at a time.  Instead, the dictionary is merged into a single string by concatenating the patterns with a unique delimiter between them.
We used the Succinct Data Structures Library (SDSL)\footnote{
\url{http://simongog.github.com/sdsl/}
} since it provides a C++ implementation of a variety of compressed suffix tree representations and it was proven to be more efficient than previous compressed suffix tree implementations \cite{GogThesis}.

	\begin{algorithm}
	\caption{Announcing Pattern Occurrence in CST}
	\label{alg:CSTocc}
	\begin{algorithmic}
		\IF{\textit{getCharAtNodePos(curNode, curNodeIndex)} = END\_OF\_STRING\_MARKER}

	    \STATE {\textit{pos} $\gets csa[lb(curNode)]-1$} 
	    \STATE \COMMENT{$lb(v)$ returns the left bound of node $v$ in the suffix array}
  	  \STATE \COMMENT{\textit{pos} is dictionary index immediately preceding this leaf's ancestor emanating from \textit{root}}
    	\IF {\textit{pos}$<0$}
    			\STATE{\textit{occ} $\gets$ \textit{true}} \COMMENT{beginning of first pattern}
    	\ELSE
    			\STATE {$c \gets$ \textit{getCharAtPatternPos(pos)}}
	        \IF {$c = $ END\_OF\_STRING\_MARKER}
            \STATE{\textit{occ} $\gets$ \textit{true}} \COMMENT{beginning of some pattern after first}
					\ENDIF
    	\ENDIF
		\ENDIF
	\end{algorithmic}
	\end{algorithm} 
	\normalsize

Although the ultimate capability of the compressed suffix tree  is modeled after the functionality of its uncompressed counterpart, many operations that are straightforward in the uncompressed suffix tree require creativity in the compressed data structures.  Understanding how the suffix tree components are represented in the compressed variation is a necessary prerequisite to implementing seemingly straightforward navigational tasks.  Furthermore, the compressed suffix tree is a self-index and allows us to discard the original set of patterns.  Thus, we had to figure out which component data structure to query in order to randomly access a single pattern character.
For instance, announcing a pattern occurrence (Algorithm \ref{alg:DictST}, line \ref{line:announce}) is not simply a question of checking whether traversal has reached the end of a leaf representing the first suffix of a pattern.
A simple \emph{if} statement is replaced by the segment of pseudocode delineated in Algorithm \ref{alg:CSTocc} and described in the following paragraph.

Instead of an \emph{if} statement that checks properties of a leaf, we perform the following computation, involving several function calls, to determine if a pattern occurrence has been located in the text.
When traversing the compressed suffix tree according to the text, a mismatch along an edge leading into a leaf may in fact be a pattern occurrence.  Thus, we first check if the mismatch is a string delimiter, which mismatches every text character.
Then, we determine if this leaf represents the first suffix of some pattern. 
This is done by finding out which character precedes the beginning of this leaf's path from the root.
If the path begins at the beginning of the dictionary, this leaf represents the first suffix of the first pattern, and a pattern occurrence is announced.
Similarly, if the character at that position is a pattern delimiter, the suffix is a complete pattern, and a pattern occurrence is announced.

The skip-count trick we described in the previous section enables us to navigate the compressed suffix tree while processing the text in linear time.
When we use this technique and traverse suffix links to find pattern occurrences in the text, some pattern occurrences can pass unnoticed.  This concern is limited to a dictionary in which one pattern is a proper substring of another.  Consider the suffix tree in Figure \ref{fig:CST_LMA} for the dictionary of patterns \texttt{D=\{a, ate, bath, later\}}.  Two of the patterns in the dictionary are substrings of other patterns.
If the text contains the word \texttt{lately}, an occurrence of the pattern \texttt{ate} should be identified within this word.  
However, using suffix links, we navigate from the node labeled \texttt{later} to the node labeled \texttt{ater} to the node labeled \texttt{ter}, without recognizing an occurrence of \texttt{ate}.  This is because we are looking for the longer pattern \texttt{later}.


	\begin{figure}[tb]
	\begin{center}
		\includegraphics[scale=.3]{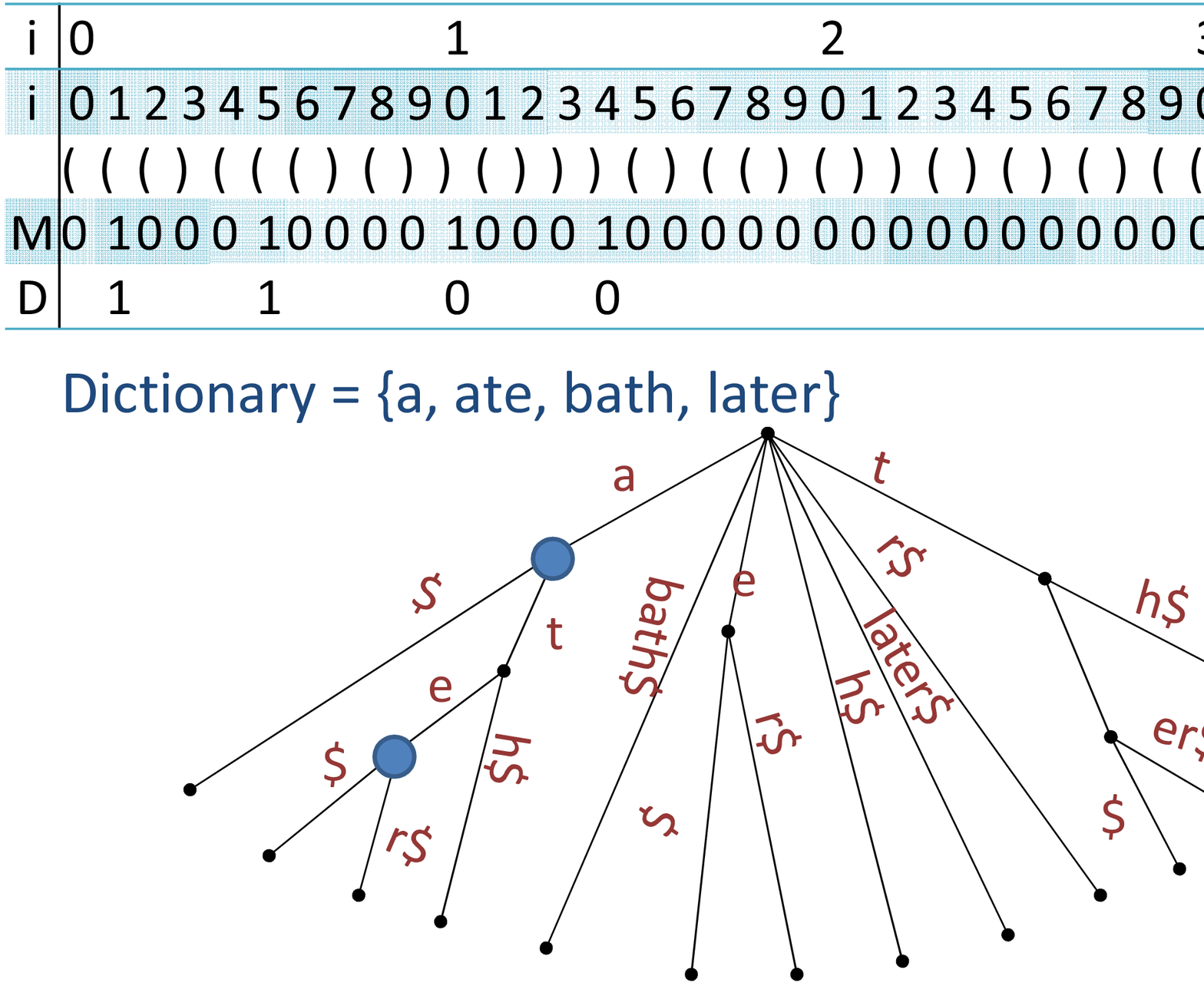}
		\caption{Suffix tree for a dictionary in which two patterns are proper substrings of other patterns.  Two nodes are marked.  A depth-first search is performed on the nodes to set up arrays M and D, depicted above the suffix tree along with the balanced parentheses representation of the tree structure.}
		\label{fig:CST_LMA}
	\end{center}
	\end{figure}

In the uncompressed suffix tree, we mark nodes that are pattern occurrences and preprocess the suffix tree with a depth-first traversal so that lowest marked ancestor (LMA) queries can be answered in constant time.  Then, an LMA query at each traversal of a suffix link ensures that no pattern occurrence is skipped over by the skip-count trick.  In a compressed suffix tree, this is not as straightforward since the nodes are not stored as independent entities.  Thus, we implemented a framework for answering lowest marked ancestor queries in constant time that consists of bit arrays and sequences of balanced parentheses.  We coded this framework with the compressed suffix tree in mind. Yet, it is suitable for any compressed representation of an ordered tree that represents the nodes as a sequence of balanced parentheses.  This is a more general contribution of this project.

	\begin{algorithm}
	\small
	\caption{Lowest Marked Ancestor Query on node in CST}
	\label{alg:LMAonCST}
	\begin{algorithmic}
	
	 \STATE \COMMENT{returns root if node has no marked ancestor}
	 \STATE \COMMENT{rank and select queries assume that the bit-array is 0-based}
	 \STATE{}
	 
		\IF{M[node]=1} 
				\STATE return node \COMMENT{node is marked}
		\ELSE
				\STATE $pre\_y \gets$ M.rank(node+1)
				\IF{$pre\_y=0$}				          
           \STATE return root
        \ELSE
            \STATE $y \gets$ M.select$(pre\_y)-1$ 
             \IF{B[$y$]=1} 
                   \STATE return $y$   \COMMENT{ $y$ is the LMA since B[$y$]=`('}                             
             \ELSE 
                    \STATE $y1 \gets$ M.rank$(y)$ \COMMENT{coresponding index in D}
                    \STATE $y2 \gets$ D.find\_open$(y1)$
                    \STATE $y3 \gets$ D.enclose$(y2)$ 
                  \IF{$y3$ = NULL} 
                      \STATE  return root \COMMENT{no enclosing parentheses }                   
                 \ELSE
                       \STATE $y4 \gets$ M.select$(y3+1)-1$ \COMMENT{map from D to M}
                       \STATE {return $y4$}
                \ENDIF
            \ENDIF
				\ENDIF
		\ENDIF
	\end{algorithmic}
	\end{algorithm} 
	\normalsize

We built a succinct framework that answers LMA queries in constant time by augmenting the compressed suffix tree with a bit-array, M, and a sequence of balanced parentheses, D. M and D are populated by a depth-first traversal once the compressed suffix tree is fully constructed.
 The bit array M stores two bits per suffix tree node.  The suffix tree is traversed in depth-first order and a 1 is stored in each bit that represents a marked node.  The sequence of balanced parentheses D denotes the relationship between the marked nodes in the suffix tree, also in depth-first search order.  D is stored as a bit-array, with two bits per marked node, in which 1 stands for `(' and 0 stands for `)'.  We use bit array B to refer to the balanced parentheses representation of the nodes in the compressed suffix tree.  

The first step in performing an  LMA query on a node $x$ in the compressed suffix tree is to find out if $x$ is marked in M.  If the node is marked, it is its own LMA.  If it is not a marked node, we locate the closest marked bit to the left of the node in M, which we call $y$.  If $y$ represents the first visit to a node, an open parenthesis in B, $y$ represents the lowest marked ancestor of $x$.  Otherwise, the lowest marked ancestor of $x$ corresponds to the closest marked bit enclosing $y$.  To find the lowest marked ancestor in this case, we map $y$ from M to D and find the first open parenthesis that precedes its open parenthesis in D.  This procedure is delineated in Algorithm \ref{alg:LMAonCST}.

We refer to the suffix tree in Figure \ref{fig:CST_LMA} for illustrative examples.  
The LMA of the node labeled \texttt{a} is itself since its bit, M[1], is marked in M.  
The LMA of the node labeled \texttt{ater}, represented by M[8], is the node labeled \texttt{ate}, which corresponds to $y=5$, since B[5] is an open parenthesis.  
The LMA of the node labeled \texttt{ath}, represented by M[11], is the node labeled \texttt{a}, which corresponds to M[1], since $y=10$, B[10] is a close parenthesis in, and position 1 is the open parenthesis of the node that encloses M[10] in D. 

Algorithm \ref{alg:LMAonCST} performs constant-time lowest marked ancestor (LMA) queries on the compressed suffix tree and consists of a set of  operations on bit arrays and sequences of balanced parentheses. We use data structures that answer rank \cite{Vigna08} and select \cite{ClaMun96} queries on bit arrays in constant time and find\_open, find\_close, and enclose queries on sequences of balanced parentheses in constant time \cite{SadNav10}.
A rank query, rank($i$), returns the number of 1's in the first $i$  positions of the array.
A select query, select($i$), finds the position of the $i$th 1 in the bit array.
find\_open($i$) and find\_close($i$)  queries find the matching parenthesis for the  parenthesis at position $i$. enclose($i$) finds the closest enclosing pair of parentheses to the parenthesis at position $i$.
We used efficient implementations of these data structures that are included in the Succinct Data Structures Library. 

	
\section{Experimental Results}\label{sec:experiments}
We implemented the algorithms in C++ and ran experiments on computers that feature an Intel(R) Xeon(R) processor at 2.93 GHz, with 5 GB of RAM, running Linux kernel  version 2.6.32.
For one set of experiments,
we searched a 5 MB English text for common English words in a 2.5 MB dictionary that comes from  ClueWeb09\footnote{\url{http://lemurproject.org/clueweb09/}} and  was used in \cite{Boy11}.
We performed another set of experiments on biological data.  
We searched 5 MB of the human genome for patterns in a 3 MB dictionary of  promoter sequences in the human genome\footnote{\url{http://epd.vital-it.ch}}.  
The texts are 5 MB of DNA and 5 MB of English text from the Pizza\&Chili corpus\footnote{\url{http://pizzachili.dcc.uchile.cl}}. 




We used the framework of Sadakane's compressed suffix tree \cite{Sad07}, cst\_sada, in our experiments since it stores nodes as a sequence of balanced parentheses and we were able to augment it for constant-time lowest marked ancestor queries.
Configurations of cst\_sada with newer representations of its components beat
the runtime of configurations of the other types of compressed suffix trees  on almost all operations and its space savings is of comparable significance \cite{GogThesis}.  In particular, the navigational operations are very fast.
Sadakane's CST consists of a compressed suffix array, a compressed LCP array, and a navigational structure.
We ran experiments on four different variations of Sadakane's CST using two different types of compressed suffix arrays, csa\_sada and csa\_wt, and two different types of compressed LCP arrays, lcp\_dac and lcp\_support\_tree2.  We also ran our experiments on an uncompressed suffix tree.
The  csa\_sada class is a very clean reimplementation of Sadakane's compressed suffix array  \cite{Sad02} and the csa\_wt class is based on a wavelet tree.
The lcp\_dac class uses the direct accessible code solution of Brisaboa et al.\ \cite{BriLadNav13}, which represents the LCP array in suffix array order,
and lcp\_support\_tree2 uses a tree compressed representation of the LCP array, which is based on the topology of the compressed suffix tree.

We compare the time-space trade-off of dictionary matching  using different variations of 1D dictionary matching software.
For a baseline, we use uncompressed components, which consume the most space but perform operations in constant time.  
The remaining runs use different underlying representations of compressed suffix arrays and compressed LCP arrays as components.  

Compressed suffix trees conserve a considerable amount of space while  the sacrifice is a negligible slowdown in running time.  This is illustrated in Figure \ref{fig:expGraphs} and in Tables \ref{table:expResultsDNA} and \ref{table:expResultsEng}. 



	\begin{figure}[tb]
	\begin{center}
		\includegraphics[scale=.4]{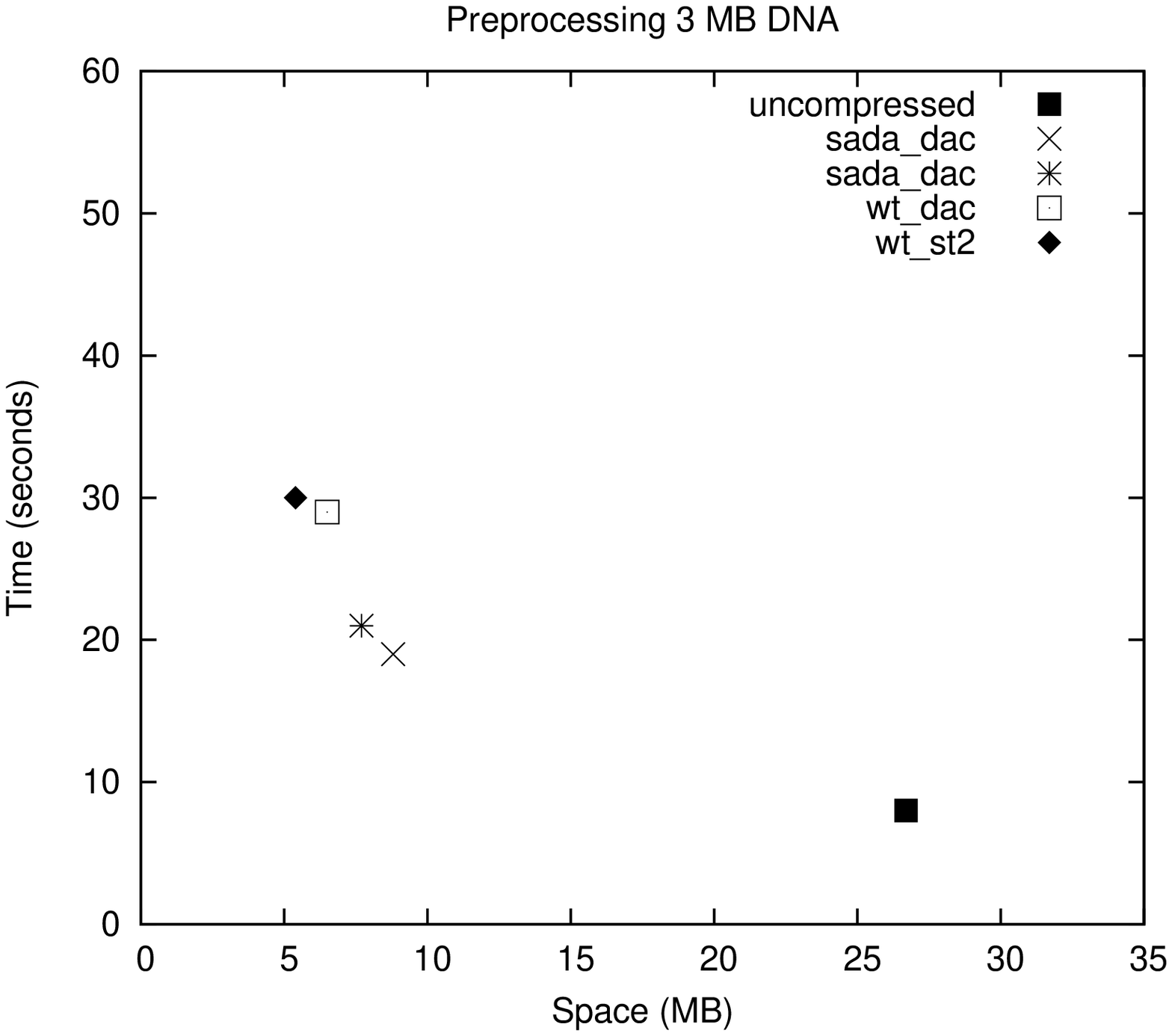}
		\hspace{10pt}
		\includegraphics[scale=.4]{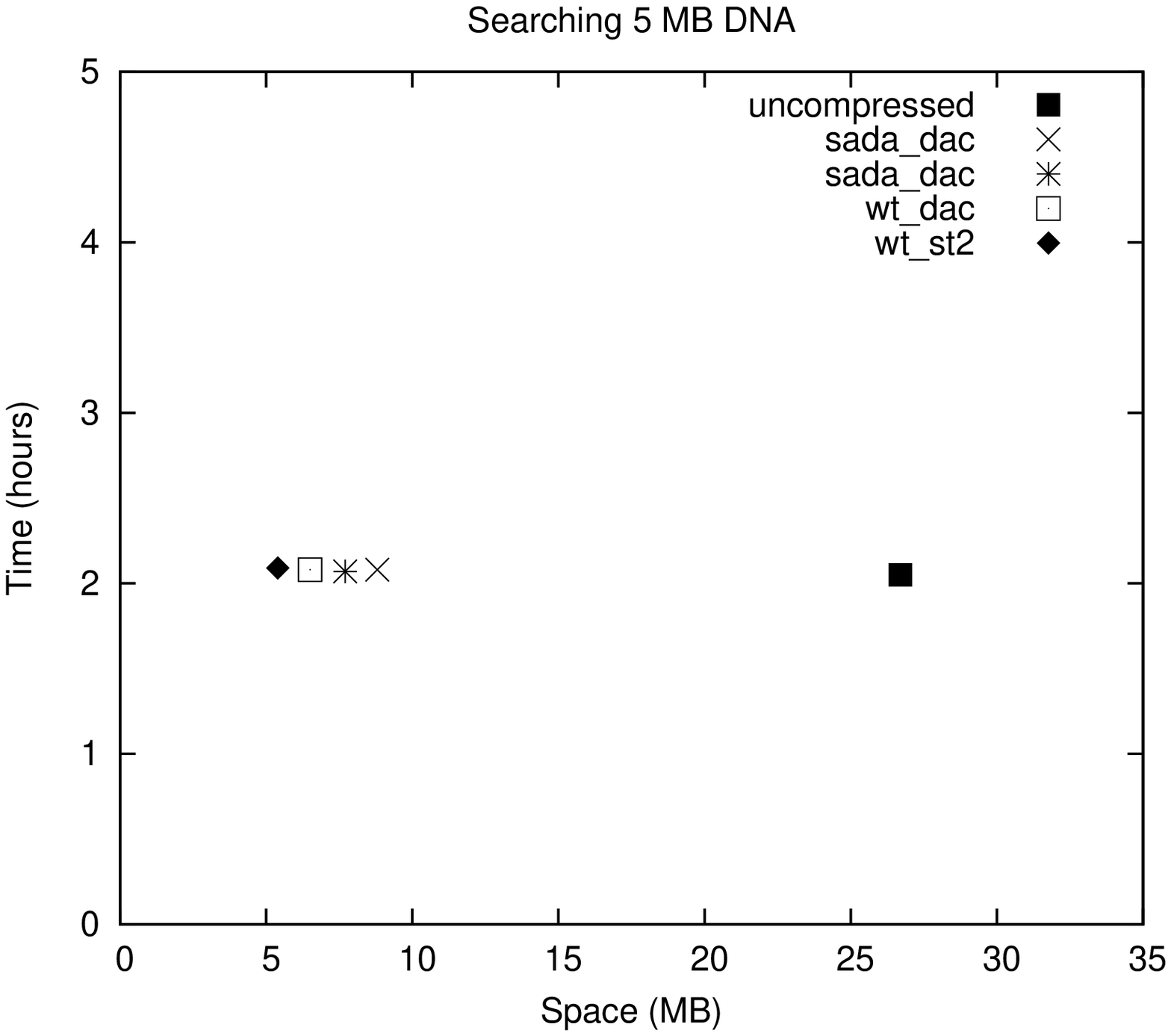}
		
		\includegraphics[scale=.4]{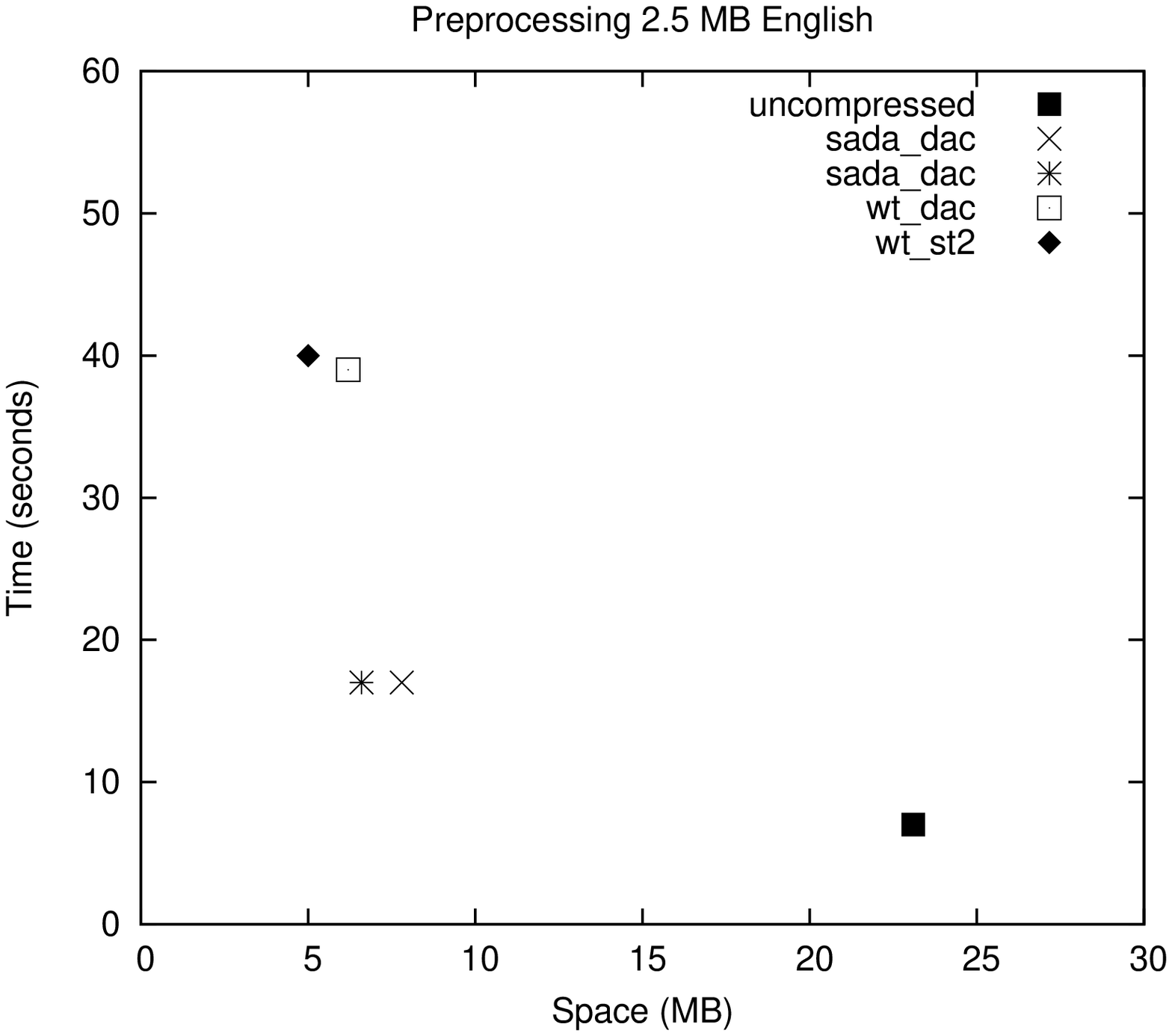}
		\hspace{10pt}
		\includegraphics[scale=.4]{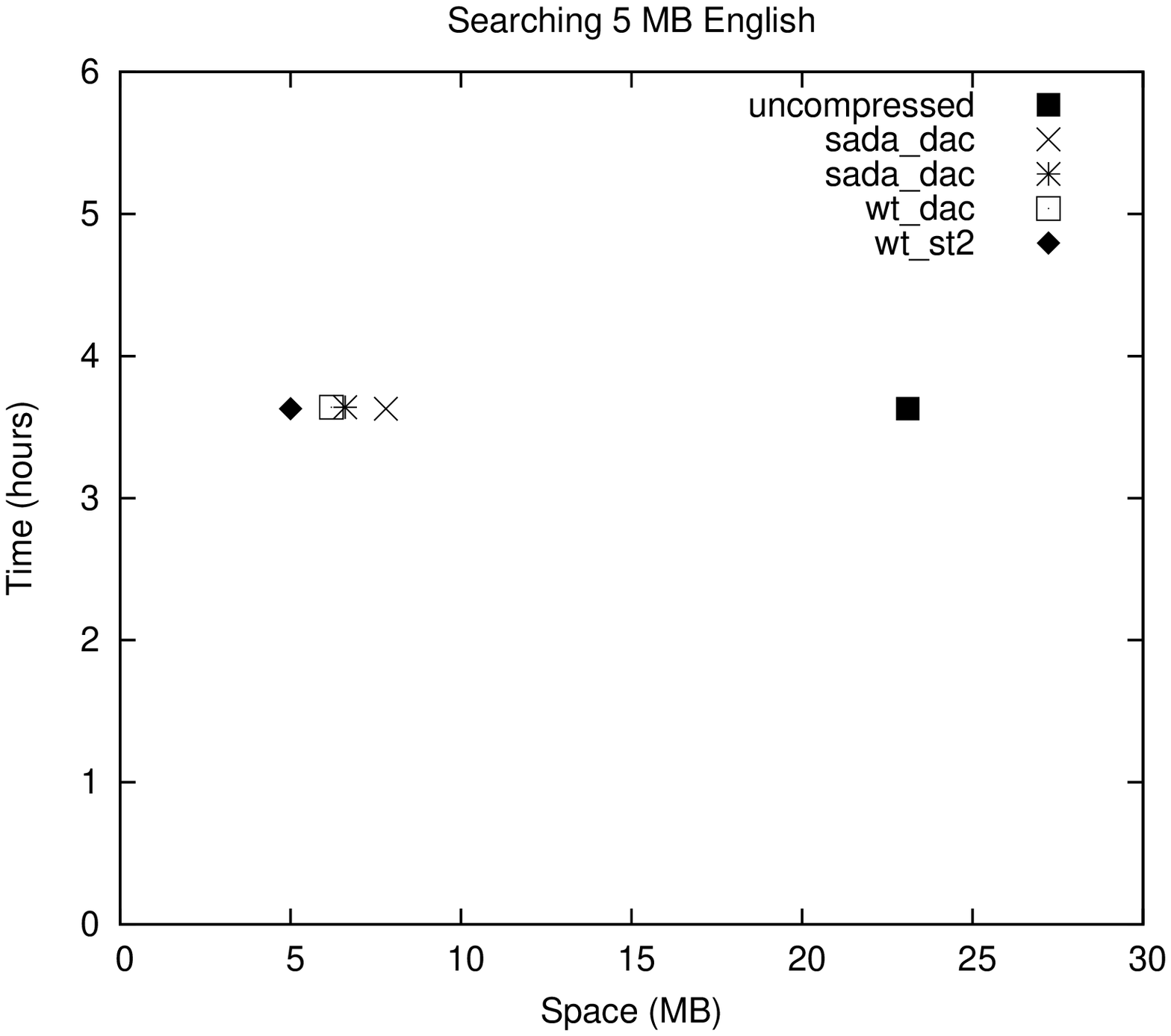}

		\caption{Time-space trade-offs of different representations of the compressed suffix array and the compressed LCP array in Sadakane's compressed suffix tree for dictionary matching.  The uncompressed is the outlier in space consumption.  The four compressed versions have similar time and space complexities.}
		\label{fig:expGraphs}
	\end{center}
	\end{figure}
\begin{table}
	\centering
		\begin{tabular}{|l|l|l|l|}
			\hline
			\multicolumn{4}{|c|}{\textbf{Experiments on DNA: 3 MB dictionary and 5 MB text}} \\
			\hline
			\textbf{CST components} & \textbf{Space} & \textbf{Preprocessing Time} & \textbf{Searching Time} \\ 
			\hline
uncompressed & 26.7 MB & 		8 sec  & 2.05 hours \\ 
\hline
sada\_dac & 8.8 MB &         19 sec    &  2.08 hours \\ 
\hline
sada\_st2 &  7.7 MB   &       21 sec  &   2.07 hours\\ 
\hline
wt\_dac & 6.5 MB     &     29 sec  &     2.08 hours \\ 
\hline
wt\_st2 & 5.4 MB   &       30 sec   &   2.09 hours \\ 
			\hline
		\end{tabular}
		\caption{Time-space trade-offs of using different representations of the compressed suffix array and the compressed LCP array in Sadakane's compressed suffix tree for searching 5 MB of the human genome for promoter sequences that comprise a 3 MB dictionary.  13 pattern occurrences were found in the text.
		} 
		\label{table:expResultsDNA}

		\begin{tabular}{|l|l|l|l|}
			\hline
			\multicolumn{4}{|c|}{\textbf{Experiments on English text: 2.5 MB dictionary and 5 MB text}} \\
			\hline
			\textbf{CST components} & \textbf{Space} & \textbf{Preprocessing Time} & \textbf{Searching Time} \\  
			\hline
uncompressed & 23.1 MB & 7 sec  & 3.63 hours \\ 
\hline
sada\_dac & 7.8 MB   &       17 sec    & 3.63 hours \\ 
\hline
sada\_st2 & 6.6 MB &         17 sec &     3.64 hours\\ 
\hline
wt\_dac & 6.2 MB     &     39  sec  &   3.64 hours \\ 
\hline
wt\_st2 & 5.0 MB     &     40  sec &     3.63 hours\\ 
			\hline
		\end{tabular}
		\caption{Time-space trade-offs of using different representations of the compressed suffix array and the compressed LCP array in Sadakane's compressed suffix tree for searching 5 MB of English text for common English words in a 2.5 MB dictionary.  12,717 pattern occurrences were located in the text.}		 
		\label{table:expResultsEng}
\end{table}

\section{Conclusion}\label{sec:conclusion}

We have introduced dictionary matching software that runs in small space.  Its underlying data structure is the compressed suffix tree.  This program runs in linear time, disregarding the slowdown of querying the compressed self-index.  We have shown that our implementation conserves considerable space in practice.  Our software includes a space-efficient technique for performing lowest marked ancestor queries on compressed suffix trees, a contribution that is useful for many other applications.

We would like to extend our small-space dictionary matching software to accommodate a dynamically changing set of patterns in the dictionary.  Several dynamic compressed suffix tree representations have been presented \cite{ChanHLS07, RusNavOli08} but they lack implementations.  Extending this work to the dynamic setting would begin by implementing the dynamic compressed suffix tree to accommodate insertion, deletion, and modification of dictionary patterns, without rebuilding the index of the entire dictionary.

\subsection*{Acknowledgements}
The authors would like to thank Simon Gog for his help installing and working with the Succinct Data Structures Library.

\small

\end{document}